\documentclass[11pt]{article}
\usepackage{times,amsthm,amsmath,amssymb}
\usepackage[legalpaper, margin=1in]{geometry}

\newtheorem{lemma}{Lemma}
\newtheorem{theorem}{Theorem}

\usepackage{tikz}
\usepackage{float}
\begin{document}

\title{A Tight Lower Bound for Clock Synchronization \\
       in Odd-Ary M-Toroids
      }

\author{Reginald Frank\footnote{Supported in part by the Distributed
    Research Experiences for Undergraduates (DREU) program, a joint
    project of the CRA Committee on the Status of Women in Computing
    Research (CRA-W) and the Coalition to Diversify Computing (CDC),
    which is funded in part by the NSF Broadening Participation in
    Computing program (NSF CNS-0540631).}\\
       Texas A\&M University, USA \\
       reginaldfrank77@tamu.edu \\
       \and
       Jennifer L. Welch\footnote{Supported in part by NSF grant 1526725.} \\
       Texas A\&M University, USA \\
       welch@cse.tamu.edu \\
}

\maketitle

\begin{abstract}

Synchronizing clocks in a distributed system in which processes
communicate through messages with uncertain delays is subject to
inherent errors. Prior work has shown upper and lower bounds on the
best synchronization achievable in a variety of network topologies and
assumptions about the uncertainty on the 
message delays.
However, until now
there has not been a tight closed-form expression for the optimal
synchronization in $k$-ary $m$-cubes with wraparound, where $k$ is
odd. In this paper, we prove a lower bound of
$\frac{1}{4}um\left(k-\frac{1}{k}\right)$, where $k$ is the (odd)
number of processes in the each of the $m$ dimensions, and $u$ is the
uncertainty in delay on every link.  Our lower bound matches the
previously known upper bound.
\end{abstract}

\section{Introduction}

Synchronizing clocks in a distributed system in which processes
communicate through messages with uncertain delays is subject to
inherent errors.  A body of work has sought bounds on how closely
the clocks can be synchronized when there is no drift in the hardware
clocks and there are no failures.  
Prior work has shown upper and lower bounds on the
best synchronization achievable in a variety of network topologies and
assumptions about the uncertainty on the message delays.

Lundelius and Lynch \cite{LundeliusL1984} showed that, in an 
$n$-process
clique with the same uncertainty $u$ on every link, the best
synchronization possible is $u\left(1 - \frac{1}{n}\right)$.
Subsequently, Halpern et al. \cite{HalpernMM1985} considered arbitrary
topologies in which each link may have a different
uncertainty, and they showed that the optimal clock synchronization is
the solution of an optimization problem; however, no general
closed-form expression was given.
Biaz and Welch \cite{BiazW2001} gave a collection of closed-form upper
and lower bounds on the optimal clock synchronization for 
$k$-ary $m$-cubes ($m$-dimensional hypercubes with $k$
processes in every dimension), both with and without 
wraparound, in which
every link has the same uncertainty, $u$.  When there is no
wraparound, the tight bound is $\frac{1}{2}um\left(k-1\right)$.  When
there is wraparound and $k$ is even, the tight bound is
$\frac{1}{4}umk$.  However, when there is wraparound and $k$ is odd,
there is a gap between the upper bound of $\frac{1}{4}um\left(k -
\frac{1}{k}\right)$ and the lower bound of $\frac{1}{4}um\left(k -
1\right)$.

In this paper, we consider $k$-ary $m$-cubes with wraparound
(``$m$-toroids'') and odd $k$.  We show a lower bound of
$\frac{1}{4}um\left(k-\frac{1}{k}\right)$, which matches the previously
known upper bound.  We use the same shifting technique from previous
lower bounds for clock synchronization (e.g.,
\cite{LundeliusL1984,HalpernMM1985,BiazW2001}).  The key insight in
our improved lower bound is to exploit the fact that the graph is a
collection of rings in each dimension and to use multiple shifted 
executions instead of just one.

\section{Preliminaries}

We first present our model and problem statement (following
\cite{LundeliusL1984,AttiyaW2004,BiazW2001}).  We consider a graph of
$k^{m}$ processes, where $k \ge 3$ is odd and $m \ge 1$, 
in which each process id is a tuple
$\left\langle p_0,p_1,...,p_{m-1}\right\rangle$ 
where each $p_i \in \left\lbrace 0,1,...,k-1\right\rbrace$. 
There are links in both directions between any two processes $\vec{p}$ 
and $\vec{q}$ if and only if
their ids differ in exactly one component, say the $i$-th, such that
$p_i = q_i + 1$ (addition on process indices is modulo $k$
throughout).  Each process $\vec{p}$ has a {\em hardware clock}
modeled as a function $HC_{\vec{p}}$ from reals (real time) to reals
(clock time).  We assume there is no drift, so $HC_{\vec{p}}(t) = t +
c_{\vec{p}}$ for some constant $c_{\vec{p}}$.  Each {\em process} is
modeled as a state machine whose transition function takes as input
the current state, current value of the hardware clock, and current
{\em event} (receipt of a message or some internal occurrence), and
produces a new state and a message to send over each incident link.

A {\em history} of process $\vec{p}$ is a sequence of alternating
states and pairs of the form (event, hardware clock value), beginning
with $\vec{p}$ 's initial state.  Each state must follow correctly
from the previous one according to $\vec{p}$ 's transition function
and the hardware clock values must increase.
A {\em timed history} of $\vec{p}$ is a history together with an
assignment of a real time $t$ to each pair $(e,T)$ in the history such
that $HC_{\vec{p}}(t) = T$.  An {\em execution} is a set of $k^{m}$ timed
histories, one per process, with a bijection 
for each link between the set of messages sent over the link and the set
of messages received over the link.
The {\em delay} of a message is the difference between the
real time when it is received and the real time when it is sent.
An execution is {\em admissible} if every message has delay in 
$\left[0,u\right]$ where $u$ is a fixed value called the 
{\em uniform uncertainty}.

We assume each process $\vec{p}$ has a local variable {\em adj}$_{\vec{p}}$ 
as part of its state and we define its {\em adjusted clock} $AC_{\vec{p}}(t)$ 
to be equal to $HC_{\vec{p}}(t) + adj_{\vec{p}}(t)$.  
An execution has {\em terminated} once all processes have stopped
changing their {\em adj} variables.  We say the algorithm {\em achieves 
$\epsilon$-synchronized clocks} if every admissible execution eventually 
terminates with
$|AC_{\vec{p}}(t) - AC_{\vec{q}}(t)| \le \epsilon$ for all 
processes $\vec{p}$ and $\vec{q}$ and all times $t$ after termination.

``Shifting'' an execution changes the real times at which events occur
\cite{LundeliusL1984}.  
Let {\bf x} be an $m$-dimensional matrix of real numbers
with $k$ elements in each dimension, which we call a {\em shift matrix}; 
elements of {\bf x} are indexed by process ids.
Define {\em shift}$(\alpha,${\bf x}$)$ be the result of adding $x_{\vec{p}}$
to the real time associated with each event in $\vec{p}$ 's timed history
in $\alpha$.  Shifting changes the hardware clocks and message delays
as follows \cite{LundeliusL1984,AttiyaW2004}:

\begin{lemma} 
\label{lem:shifting}
Let $\alpha$ be an execution with hardware clocks $HC_{\vec{p}}$
and let {\bf x} be a shift matrix.
Then {\em shift}$(\alpha,${\bf x}$)$ is a (not necessarily admissible)
execution in which
\begin{itemize}
\item[(a)] the hardware clock of each $\vec{p}$, denoted $HC_{\vec{p}}'(t)$, equals
$HC_{\vec{p}}(t) - x_{\vec{p}}$ and
\item[(b)] every message from $\vec{p}$ to $\vec{q}$ has delay $\delta - x_{\vec{p}} + x_{\vec{q}}$,
where $\delta$ is the message's delay in $\alpha$.
\end{itemize}
\end{lemma}

\section{Lower Bound}

\begin{theorem}
	\label{theorem:Lower Bound}
For any algorithm that achieves $\epsilon$-synchronized clocks in a
$k$-ary $m$-toroid with uniform uncertainty $u$, where $k$ is odd, it
must be that $\epsilon \ge \frac{1}{4}um\left(k-\frac{1}{k}\right)$.
\end{theorem}

\begin{proof}
Let $\cal A$ be any algorithm that achieves $\epsilon$-synchronized clocks in
a $k$-ary $m$-toroid with uniform uncertainty $u$,
where $k = 2r+1$ for some integer $r \ge 1$. 
Let $\alpha$ be the admissible execution of $\cal A$ in which 
$HC_{\vec{p}}(t) = t$ for each process $\vec{p}$,
every message from $\vec{p}$ to $\vec{q}$, 
where $\vec{q}$ is $\vec{p}$'s neighbor in the $h$-th 
dimension such that
$q_h = p_h + 1$,
has the same fixed delay 
$\delta_{\vec{p},\vec{q}}$, which is $0$ if $0 \le p_h < r$ and is $u$ if
$r \le p_h < k$,
and every message from $\vec{q}$ to $\vec{p}$ has the same fixed delay
$\delta_{\vec{q},\vec{p}} = u - \delta_{\vec{p},\vec{q}}$.

For $0 \le i < k$, define $\alpha^i =$ {\em shift}$(\alpha,${\bf x}$^i)$,
where the $\vec{p}$-th element of the shift matrix
{\bf x}$^i$, denoted 
$x^i_{\vec{p}}$, is defined as $\sum _{j=0}^{m-1}{{\bf W}^i_{p_j}},$ 
where {\bf W} is defined as follows:

\begin{center}
\begin{tabular}{|c|c||c|c|}
\hline
\multicolumn{4}{|c|}{range of $i$} \\
\hline
\hline
\multicolumn{2}{|c||}{$0 \le i < r$} &
\multicolumn{2}{c|} {$r \le i < k$} \\
\hline
range of $p_j$ & ${{\bf W}^i_{p_j}}$ & range of ${p_j}$ & ${{\bf W}^i_{p_j}}$ \\
\hline
\hline
$0 \le {p_j} \le i$       & 0           & $0 \le {p_j} \le i-r$ & ${p_j}u$     \\
\hline
$i < {p_j} \le r$         & $({p_j}-i)u$    & $i-r < {p_j} \le r$   & $(i-r)u$ \\
\hline
$r < {p_j} \le r + i + 1$ & $(r-i)u$    & $r < {p_j} \le i$     & $(i-{p_j})u$ \\
\hline
$r+i +1 < {p_j} \le 2r$   & $(2r-{p_j}+1)u$ & $i < {p_j} \le 2r$    & $0$      \\
\hline
\end{tabular}
\end{center}

The idea behind the shift amounts in ${\bf W}$ is
to cause two processes that are farthest apart in the
graph to be shifted as far apart in real time as possible---thus achieving 
a large skew between their
adjusted clocks---while maintaining valid message delays between all
neighbors.  By considering multiple shifted executions, we can
cancel out terms involving adjusted clocks, leaving behind only
terms that involve the system parameters $\epsilon$ and $u$, and the
graph parameters $k$ and $m$.

As an example, consider the case when $k = 5$ and $m = 1$, that is, a
5-element ring (cf.\ Figure \ref{fig:Graph of a 5-ary 1-toroid}).  We will
denote the process with id $\langle i \rangle$ by $p_i$.

\begin{figure}[H]
	\begin{center}
		\centering
		\begin{minipage}{.49\textwidth}
			\centering
			\begin{tikzpicture}[scale=.8, every node/.style={scale=.8}]
			\def \radius {3cm}
			\def \Bradius {\radius+3}
			\def \BLradius {\Bradius+7}
			\def \Sradius {\radius-3}
			\def \SLradius {\Sradius-7}
			\def \margin {8}
			
			\foreach \n in {0,1,2,3,4}
			{
				\node[draw, circle] at ({72*\n+90}:\radius) {$p_\n$};
				\draw[<->, >=latex] ({72*\n+\margin+90}:\radius) arc ({72*\n+\margin+90}:{90+72 * (\n+1)-\margin}:\radius);
			}
			\end{tikzpicture}
			\caption{Graph of a 5-ary 1-toroid}
			\label{fig:Graph of a 5-ary 1-toroid}
		\end{minipage}

		\begin{minipage}{.49\textwidth}
			\centering
			\begin{tikzpicture}[scale=1.1, every node/.style={scale=.775}]
			\foreach \y in {0,.7,1.4,2.1,2.8,3.5}{
				\draw[very thick,->] (0,\y) -- (3,\y);
			}
			\foreach \y in {0,1,2,3,4}
			\draw (-.5,3.5-\y*.7) -- (-.5,3.5-\y*.7) node[anchor=east] {$p_\y$}; 
			\draw (-.5,0) -- (-.5,0) node[anchor=east] {$p_0$};
			
			\draw[very thick,densely dotted,gray] (.02,-.25) -- (.02,3.75);
			\draw (.02,4) -- (.02,4) node {$0$};
			\draw[very thick,densely dotted,gray] (1.77,-.25) -- (1.77,3.75);
			\draw (1.77,4) -- (1.77,4) node {$u$};
			
			\foreach \y in {.7,1.4,2.1}{
				\draw[thick,->,>=latex] (.02,\y) -- (1.8,\y-.7);
				\draw[thick,->,>=latex] (.02,\y-.7) -- (.02,\y);
			}
			\foreach \y in {2.8,3.5}{
				\draw[thick,->,>=latex] (.02,\y-.7) -- (1.77,\y);
				\draw[thick,->,>=latex] (.02,\y) -- (.02,\y-.7);
			}
			\foreach \y in {0,.7,1.4,2.1,2.8,3.5}{
				\draw[very thick] (0,\y) -- (1,\y);
			}
			\draw[thick,->,white] (0,.4-.78)--(1,.4-.78);
			\draw[thick,->,white] (0,4.5-.078)--(.1,4.5-.078);
			\end{tikzpicture}
			\caption{Delay Pattern in $\alpha$ for 5-ary 1-toroid}
			\label{fig:Defined Delay Pattern for 5-ary 1-toroid}
		\end{minipage}
	\end{center}
\end{figure}

Figure \ref{fig:Defined Delay Pattern for 5-ary 1-toroid} depicts the
pattern of message delays in $\alpha$ for the 5-element ring, with $p_0$
occurring twice for convenience in representing the wrap-around.  The
interpretation is that every message, if any, sent from $p_0$ to $p_1$
has delay 0, every message sent from $p_1$ to $p_0$ has delay $u$, etc.
We make no assumption about when or if such messages are sent, as that
depends on the algorithm.

\begin{figure}[H]
	\begin{center}
		\centering
		\begin{minipage}{.48\textwidth}
			\centering
			\begin{tikzpicture}[scale=0.775, every node/.style={scale=0.775}]
			
			\foreach \y in {0,3.5,2.8}{
				
				\draw[very thick,->] (0,\y) -- (4.25,\y);
			}
			\foreach \y in {.7,1.4,2.1}{
				\draw[very thick,->] (1.75,\y) -- (6,\y);
			}
			\foreach \y in {0,1,2,3,4}
			\draw (-.5,3.5-\y*.7) -- (-.5,3.5-\y*.7) node[anchor=east] {$p_\y$}; 
			\draw (-.5,0) -- (-.5,0) node[anchor=east] {$p_0$};
			
			\draw[very thick,densely dotted,gray] (.02,-.25) -- (.02,3.75);
			\draw (.02,4) node {$0$};
			\draw[very thick,densely dotted,gray] (1.77,-.25) -- (1.77,3.75);
			\draw (1.77,4) node {$u$};
			\draw[very thick,densely dotted,gray] (3.52,-.25) -- (3.52,3.75);
			\draw (3.52,4) node {$2u$};

			\foreach \y in {1.4,2.1}{
				\draw[thick,->,>=latex] (1.77,\y) -- (1.8+1.75,\y-.7);
				\draw[thick,->,>=latex] (1.77,\y-.7) -- (1.77,\y);
			}
			\foreach \y in {3.5}{
				\draw[thick,->,>=latex] (.02,\y-.7) -- (1.77,\y);
				\draw[thick,->,>=latex] (.02,\y) -- (.02,\y-.7);
			}
			\foreach \y in {2.8}{
				\draw[thick,->,>=latex] (.02+\y*-1*1.75/.7+3.5*1.75/.7-1.75,\y) -- (1.8+\y*-1*1.75/.7+3.5*1.75/.7-1.75,\y-.7);
				\draw[thick,->,>=latex] (1.77,\y-.7) -- (1.77+\y*-1*1.75/.7+3.5*1.75/.7-1.75,\y);
			}
			
			\draw[thick,->,>=latex] (1.77,.7) -- (1.77,0);
			\draw[thick,->,>=latex] (.02,0) -- (1.77,.7);
			\foreach \y in {0,3.5,2.8}{
				\draw[very thick] (-.005,\y) -- (1,\y);	
			}
			\foreach \y in {.7,1.4,2.1}{

				\draw[very thick] (1.75-.005,\y) -- (1.75-.005+1,\y);
			}
			\end{tikzpicture}
			\caption{Shifted Delay Pattern for $\alpha^1$}
			\label{fig:Shifted Delay Pattern for alpha1}
		\end{minipage}

		\begin{minipage}{.5\textwidth}
			\centering
			\begin{tikzpicture}[scale=0.775, every node/.style={scale=0.775}]
			
			\foreach \y in {2.1,2.8,3.5}{

				\draw[very thick,->] (\y*-1*1.75/.7+3.5*1.75/.7,\y) -- (4.75+\y*-1*1.75/.7+3.5*1.75/.7-.5,\y);	
			}
			\foreach \y in {.7,1.4}{

				\draw[very thick,->] (-\y*-1*1.75/.7-3.5*1.75/.7+1.75*4,\y) -- (4.75-\y*-1*1.75/.7-3.5*1.75/.7+1.75*4-.5,\y);
			}
			\draw[very thick,->] (0,0) -- (4.25,0);
			\foreach \y in {0,1,2,3,4}
			\draw (-.5,3.5-\y*.7) -- (-.5,3.5-\y*.7) node[anchor=east] {$p_\y$}; 
			\draw (-.5,0) -- (-.5,0) node[anchor=east] {$p_0$};
			
			\draw[very thick,densely dotted,gray] (.02,-.25) -- (.02,3.75);
			\draw (.02,4) node {$0$};
			\draw[very thick,densely dotted,gray] (1.77,-.25) -- (1.77,3.75);
			\draw (1.77,4) node {$u$};
			\draw[very thick,densely dotted,gray] (3.52,-.25) -- (3.52,3.75);
			\draw (3.52,4) node {$2u$};
			
			\foreach \y in {3.5,2.8}{
				\draw[thick,->,>=latex] (.02+\y*-1*1.75/.7+3.5*1.75/.7,\y) -- (1.77+\y*-1*1.75/.7+3.5*1.75/.7,\y-.7);
				\draw[thick,->,>=latex] (1.77+\y*-1*1.75/.7+3.5*1.75/.7,\y-.7) -- (1.77+\y*-1*1.75/.7+3.5*1.75/.7,\y-.01);
			}
			\foreach \y in {2.1,1.4}{
				\draw[thick,->,>=latex] (.02-\y*-1*1.75/.7-3.5*1.75/.7+4*1.75,\y) -- (.02-\y*-1*1.75/.7-3.5*1.75/.7+4*1.75,\y-.7);
				\draw[thick,->,>=latex] (.02-\y*-1*1.75/.7-3.5*1.75/.7+1.75*3,\y-.7) -- (.02-\y*-1*1.75/.7-3.5*1.75/.7+1.75*4,\y+.01);
			}
			\draw[thick,->,>=latex] (.02,.7) -- (1.8,0);
			\draw[thick,->,>=latex] (.02,0) -- (.02,.7);
			\foreach \y in {2.1,2.8,3.5}{
				\draw[very thick] (\y*-1*1.75/.7+3.5*1.75/.7-.005,\y) -- (1+\y*-1*1.75/.7+3.5*1.75/.7-.5,\y);	
			}
			\foreach \y in {.7,1.4}{
				\draw[very thick] (-\y*-1*1.75/.7-3.5*1.75/.7+1.75*4-.005,\y) -- (1-\y*-1*1.75/.7-3.5*1.75/.7+1.75*4-.5,\y);	
			}
			\draw[very thick] (0,0) -- (1,0);
			\end{tikzpicture}
			\caption{Shifted Delay Pattern for $\alpha^4$}
			\label{fig:Shifted Delay Pattern for alpha4}
		\end{minipage}
	\end{center}
\end{figure}

Now we consider two of the five shifts for this special case. The shift
matrix defined by ${\bf W}$ for $\alpha^1$ is $[0,0,u,u,u]$
and that for $\alpha^4$ is $[0,u,2u,u,0]$.
Figures \ref{fig:Shifted Delay Pattern for alpha1} and
\ref{fig:Shifted Delay Pattern for alpha4} depict the pattern of message
delays in $\alpha^1$ and $\alpha^4$ respectively, reflecting the changes
indicated by Lemma \ref{lem:shifting}(b).
Visual inspection shows that the delays are still in the valid range
and thus the shifted executions are admissible.

Admissibility and Lemma \ref{lem:shifting}(a) imply that 
$AC'_1 - AC'_4 = AC_1 - AC_4 + u \le \epsilon$ 
and
$AC'_4 - AC'_2 = AC_4 - AC_2 + 2u \le \epsilon$.
Similarly, one can check that $\alpha^0$, $\alpha^2$, and $\alpha^3$
are admissible and then get similar inequalities.
Summing the five inequalities results in $6u \le 5\epsilon$, or
$\epsilon \ge 6u/5$ which agrees with Theorem \ref{theorem:Lower Bound}.

We now show that all shifted executions are admissible.

\begin{lemma}
\label{lem:admiss}
For all $i$, $0 \le i < k$, $\alpha^i$ is admissible.
\end{lemma}

\begin{proof}
Fix $i$ with $0 \le i < k$.
We must show that all message delays are in $[0,u]$.
Let $\vec{p}$ and $\vec{q}$ be two neighbors that differ in the $h$-th
dimension such that $q_h = p_h + 1$ and $q_j = p_j$ for all $j \ne h$.
Denote the (fixed) delay of messages from $\vec{p}$ to $\vec{q}$ 
in $\alpha^i$ by $\delta^i_{\vec{p}, \vec{q}}$.
By Lemma \ref{lem:shifting}(b), $\delta^i_{\vec{p}, \vec{q}} = \delta_{\vec{p}, \vec{q}} 
+ \Delta^i_{\vec{p},\vec{q}}$, where $\Delta^i_{\vec{p},\vec{q}}$
denotes $- x^i_{\vec{p}} + x^i_{\vec{q}}$.  
Observe that $\Delta^i_{\vec{q},\vec{p}} = - \Delta^i_{\vec{p},\vec{q}}$.

\begin{align*}
\Delta^i_{\vec{p},\vec{q}} &= - \sum_{j=0}^{m-1}{\bf W}^i_{p_j} +
   \sum_{j=0}^{m-1}{\bf W}^i_{q_j} 
       &\text{ by definition of shift vector {\bf x}$^i$ for $\alpha^i$} \\
  &= - {\bf W}^i_{p_h} + {\bf W}^i_{q_h} 
      &  \text{ since $\vec{p}$ and $\vec{q}$ only differ in the $h$-th 
                dimension} \\
  &= - {\bf W}^i_{p_h} + {\bf W}^i_{p_h + 1} 
      & \text{by definition of $\vec{q}$}
\end{align*}
Referring to the table defining {\bf W}, we get the following values
for $\Delta^i_{\vec{p},\vec{q}}$:

\begin{center}
\begin{tabular}{|c|c||c|c|}
\hline
\multicolumn{4}{|c|}{range of $i$} \\
\hline
\hline
\multicolumn{2}{|c||}{$0 \le i < r$} &
\multicolumn{2}{c|} {$r \le i < k$} \\
\hline
range of $p_h$ & 
  ${\Delta}^i_{\vec{p},\vec{q}}$ &
  range of ${p_h}$ &
  ${{\Delta}^i_{\vec{p},\vec{q}}}$ \\
\hline
\hline
$0 \le p_h < i$       & 0           & $0 \le p_h < i-r$ & $u$     \\
\hline
$i \le p_h < r$         & $u$    & $i-r \le p_h < r$   & 0 \\
\hline
$r \le p_h < r+i+1$ & 0    & $r \le p_h < i$     & $-u$ \\
\hline
$r+i+2 \le p_h \le 2r$   & $-u$ & $i \le p_h \le 2r$    & 0      \\
\hline
\end{tabular}
\end{center}

To gain an intuition for why $\alpha^i$ is admissible, consider
how the delays chosen for $\alpha$ relate to $\Delta^i_{\vec{p},\vec{q}}$.
Recall that $\vec{p}$ and $\vec{q}$ are neighbors in dimension $h$.
If $\vec{p}$ occurs before index $r$ in dimension $h$, 
then $\delta_{\vec{p},\vec{q}}$, the delay from $\vec{p}$ to $\vec{q}$ in 
$\alpha$, is chosen
so that $\Delta^i_{\vec{p},\vec{q}}$ can be maximized;
otherwise it is chosen so that $\Delta^i_{\vec{p},\vec{q}}$ can be minimized.
To keep the shifted message delays in the valid range, 
$\Delta^i_{\vec{p},\vec{q}}$ must be between $u$ and $-u$.
In particular, the delay in $\alpha$ when $\vec{p}$ occurs before
index $r$ is chosen so that $\Delta^i_{\vec{p},\vec{q}}$ can be
$u$; otherwise it is chosen so that $\Delta^i_{\vec{p},\vec{q}}$ can be
$-u$.  Below we formalize these ideas.

Since $\delta_{\vec{p},\vec{q}}$ is in $[0,u]$, so is
$\delta^i_{\vec{p},\vec{q}}$ for all table entries where
$\Delta^i_{\vec{p},\vec{q}} = 0$.  For all table entries where
$\Delta^i_{\vec{p},\vec{q}} = u$, the definition of $\alpha$ states
that $\delta_{\vec{p},\vec{q}} = 0$, and thus
$\delta^i_{\vec{p},\vec{q}} = 0 + u = u$.  For all table entries where
$\Delta^i_{\vec{p},\vec{q}} = -u$, the definition of $\alpha$ states
that $\delta_{\vec{p},\vec{q}} = u$, and thus
$\delta^i_{\vec{p},\vec{q}} = u + (-u) = 0$.  In all cases
$\delta^i_{\vec{p},\vec{q}}$ is in $[0,u]$.

Since
$\delta_{\vec{q},\vec{p}}$ is defined in $\alpha$ to be $u - \delta_{\vec{p},
\vec{q}}$ and $\Delta^i_{\vec{q},\vec{p}} = - \Delta^i_{\vec{p},\vec{q}}$,
it follows that $\delta^i_{\vec{q},\vec{p}} = u - \delta^i_{\vec{p},\vec{q}}$.
Since we just showed that $\delta^i_{\vec{p},\vec{q}}$ is in $[0,u]$, the
same is true of $\delta^i_{\vec{q},\vec{p}}$.
Thus $\alpha^i$ is admissible.
\end{proof}

Fix any $i$ with $0 \le i < r$.  We focus on two processes that are maximally
far away from each other.  Since $\alpha^i$ is admissible 
{by Lemma \ref{lem:admiss}},
$\cal A$ must ensure that
$AC^i_{\left\langle i,...,i\right\rangle } - 
AC^i_{\left\langle i+r+1,...,i+r+1\right\rangle} \le \epsilon$, where
$AC^i_{\vec{p}}$ denotes the adjusted clock of process $\vec{p}$ after 
termination in $\alpha^i$.
By definition of $\alpha^i$ and Lemma \ref{lem:shifting}(a),
$AC^i_{\left\langle i,...,i\right\rangle } = 
AC_{\left\langle i,...,i\right\rangle }$
and
$AC^i_{\left\langle i+r+1,...,i+r+1\right\rangle} =
AC_{\left\langle i+r+1,...,i+r+1\right\rangle} - m(r-i)u$.
Thus by substituting we get:
\begin{align}
AC_{\left\langle i,...,i\right\rangle } - AC_{\left\langle i+r+1,...,i+r+1\right\rangle } \le - m(r-i)u + \epsilon \mbox{ for } 0 \le i < r
\end{align}
Similarly, we can show:
\begin{align}
AC_{\left\langle i,...,i\right\rangle } - AC_{\left\langle i-r,...,i-r\right\rangle } \le - m(i-r)u + \epsilon \mbox{ for } r \le i < k
\end{align}

Adding together the $r$ inequalities from (1) and the $k-r$
inequalities from (2) gives
\begin{align}
\begin{split}
  \sum_{i=0}^{r-1} AC_{\left\langle i,...,i\right\rangle } &- \sum_{i=0}^{r-1} AC_{\left\langle i-r,...,i-r\right\rangle } +
  \sum_{i=r}^{k-1} AC_{\left\langle i,...,i\right\rangle } - \sum_{i=r}^{k-1} AC_{\left\langle i-r,...,i-r\right\rangle } \\
  &\le
  -um\left[\sum_{i=0}^{r-1}(r-i) + \sum_{i=r}^{k-1}(i-r)\right] + k\epsilon
\end{split}
\end{align}
The left-hand side of (3)
resolves to $0$ and the expression in square brackets equals $(k^2-1)/4$, and thus
$\epsilon \ge \frac{1}{4}um\left(k - \frac{1}{k}\right).$
\end{proof}

\section{Conclusion}

We have closed the gap between the best previously-known closed-form
upper and lower bounds on the optimal clock synchronization for
$k$-ary $m$-toroids when $k$ is odd and the uncertainty on each link
is the same.  By applying a more involved set of shifts than those in
the prior work \cite{BiazW2001} and exploiting the specific network
topology, we achieved a lower bound that equals the upper bound due to
the algorithm in \cite{BiazW2001}.

\end{document}